%% file: ms.tex
\documentclass{article} 

\usepackage{graphicx}
\usepackage{lineno}
\usepackage{hyperref}
\usepackage[cmex10]{amsmath}
\usepackage{amsfonts,amsthm}
\usepackage{multirow}
\usepackage{subcaption}
\usepackage{pgfplots}
\pgfplotsset{compat=newest}
\usepackage{pgfplotstable}
\usepackage{tikz}

\usepackage{listings}

\usepackage{blkarray}
\usepackage{multirow}

\usepackage{xcolor}

\usepackage[ruled,vlined,linesnumbered,commentsnumbered]{algorithm2e}

\definecolor{rankcolor}{gray}{0.6}

\usetikzlibrary{arrows.meta,shadows,calc}
\tikzset{%
	>={Latex[width=2mm,length=2mm]},
	base/.style = {rectangle, rounded corners, draw=black,
		minimum width=2cm, minimum height=1cm,
		text centered, font=\sffamily, drop shadow},
	patientCondition/.style = {base, fill=white},
	agreementChannel/.style = {base, fill=white},
	rater/.style = {base, fill=white},
	patientConditionGA/.style = {base, fill=blue!30},
	agreementChannelGA/.style = {base, fill=green!30},
	raterGA/.style = {base, fill=orange!15},	
	randomVariable/.style = {draw=none, fill=none, font=\sffamily},
}

\newcommand{\mf}{\mathfrak}
\newcommand{\mc}{\mathcal}

\newcommand{\sumCol}[2]{\textrm{S}_{#1}^{X}({#2})}
\newcommand{\sumRow}[2]{\textrm{S}_{#1}^{Y}({#2})}
\newcommand{\sumMatrix}[1]{\textrm{S}_{#1}}

\def\N{n}
\def\MI{\textit{MI}}

\def\IA{\textit{IA}}
\def\IAs{\textit{IA}_{\epsilon}}

\newcommand{\defeq}{\stackrel{\textrm{\tiny def}}{=}}

\newcommand{\genP}[2][]{p_{#1}({#2})}
\newcommand{\Px}[2][A]{\genP[X_{#1}]{#2}}
\newcommand{\Py}[2][A]{\genP[Y_{#1}]{#2}}
\newcommand{\Pxy}[2][A]{\genP[X_{#1}Y_{#1}]{#2}}
\newcommand{\trans}[1]{{#1}^{T}}

\newcommand{\eps}[2][\epsilon]{{#2}_{#1}}

\newcommand{\transeps}[2][\epsilon]{{#2}_{#1}^{T}}

\newcommand{\nnrv}[1]{\overline{#1}}

\newcommand{\Hp}[1]{H({#1})}

\newcommand\hs[2][]{\ifthenelse{\equal{#1}{}}%
	{{#2}*\ln{#2}}{{#2}*\log_{#1}{#2}}}

\newcommand{\Hr}[1]{\mathcal{G}(#1)}

\def\gAM{B}

\newcommand{\THr}[1]{\mathcal{G}^\uparrow(#1)}
\newcommand{\BHr}[1]{\mathcal{G}_\downarrow(#1)}

\modulolinenumbers[5]

\newcolumntype{x}[1]{>{\centering\let\newline\\\arraybackslash\hspace{0pt}}p{#1}}

\newtheorem{thm}{Theorem}
\newtheorem{lem}{Lemma}
\newtheorem{prop}{Proposition}
\newtheorem{cor}{Corollary}

\newtheorem{definition}{Definition}

\definecolor{codegreen}{rgb}{0,0.6,0}
\definecolor{codegray}{rgb}{0.5,0.5,0.5}
\definecolor{codepurple}{rgb}{0.58,0,0.82}
\definecolor{backcolour}{rgb}{0.95,0.95,0.92}

\lstdefinestyle{mystyle}{
	backgroundcolor=\color{backcolour},   
	commentstyle=\color{codegreen},
	keywordstyle=\color{magenta},
	numberstyle=\tiny\color{codegray},
	stringstyle=\color{codepurple},
	basicstyle=\ttfamily\footnotesize,
	breakatwhitespace=false,         
	breaklines=true,                 
	captionpos=b,                    
	keepspaces=true,                 
	numbers=left,                    
	numbersep=5pt,                  
	showspaces=false,                
	showstringspaces=false,
	showtabs=false,                  
	tabsize=2,
	language=Python
}

\title{Computing Information Agreement
	\thanks{This work has been partially supported by the ``\emph{Istituto Nazionale di Alta 
			Matematica}'' (INdAM).}
}

\makeatletter
\renewcommand\@date{{%
		\vspace{-\baselineskip}%
		\large\centering
		\begin{tabular}{@{}c@{}}
			Alberto Casagrande\textsuperscript{1} \\
			\normalsize acasagrande@units.it
		\end{tabular}%
		\quad
		\begin{tabular}{@{}c@{}}
			Francesco Fabris\textsuperscript{1} \\
			\normalsize ffabris@units.it
		\end{tabular}
		\quad 
		\begin{tabular}{@{}c@{}}
		Rossano Girometti\textsuperscript{2} \\
		\normalsize rossano.girometti@uniud.it
		\end{tabular}
		
		\bigskip
		
		\textsuperscript{1}{\normalsize{}Dipartimento di Matematica e Geoscienze,\\Universit\`a degli Studi di Trieste, Italy}\par
		\textsuperscript{2}{\normalsize{}Istituto di Radiologia, Dipartimento di Area Medica\\Universit\`a degli Studi di Udine,\\
		Ospedale S. Maria della Misericordia, Italy}
		
		\bigskip
		
		\today
}}
\makeatother

\usepackage{fancyhdr}
\begin{document}

\maketitle

\begin{abstract}
	Agreement measures are useful to both compare different evaluations of
	the same diagnostic outcomes and
	validate new rating systems or devices.
	Information Agreement (\IA{}) is an 
	information-theoretic-based agreement 
	measure introduced to overcome all the limitations and 
	alleged pitfalls of Cohen's Kappa. 
	However, it is only able to deal with agreement matrices 
	whose values are positive natural numbers. 
	This work extends \IA{} admitting also $0$ as a possible value 
	for the agreement matrix cells. 
\end{abstract}

\input{basics}
\input{extending_IA}
\input{computing}

\bibliographystyle{plain}
\bibliography{agreement-works}

\end{document}

%% file: basics.tex

\section{Basic Notions}\label{sec:basics}

Let $\mf{X}$ and $\mf{Y}$ be two raters that individually classify 
the instances of same non-empty data set $\mathcal{D}$
as belonging to one among $\N{}$ possible classes, where $\N{}$ is 
greater then $1$. 
Their combined classifications produce an \emph{agreement matrix} $A$ 
that is a $\N{}\times \N{}$-matrix whose cells $A[y][x]$ report how many 
instances of $D$ were classified, at the same 
time, as belonging to the classes $y$ and $x$ by $\mf{Y}$ and $\mf{X}$, 
respectively. 

Since $|\mathcal{D}|=\sum_{y=1}^{\N{}}\sum_{x=1}^{\N{}} A[y][x] > 0$, 
the probability for a randomly selected instance of $\mathcal{D}$ to 
be classified at the same time as belonging to the classes 
$y$ and $x$ by $\mf{Y}$ and $\mf{X}$, $\Pxy{y,x}$, equals 
$A[y][x]/\sumMatrix{A}$ where 
$\sumMatrix{A}$ is the sum of all the values in $A$, i.e., 
$\sumMatrix{A}\defeq \sum_{y=1}^{\N{}}\sum_{x=1}^{\N{}} A[y][x]$.
Since any agreement matrix contains at least one positive value, $\sumMatrix{A}$ must be greater than $0$ too.

The probability for an instance of $\mathcal{D}$ 
to be put in the class $y$ by $\mf{Y}$ is denoted by $\Py{y}$ and  
it equals $\sumRow{A}{y}/\sumMatrix{A}$, where 
$\sumRow{A}{y}$ is the sum of all the values in 
the row $y$, 
i.e., $\sumRow{A}{y}\defeq \sum_{x=1}^{\N{}} A[y][x]$.
Analogously, the odd for the same instance to be classified 
in class $x$ by $\mf{X}$ is  $\Px{x} = \sumCol{A}{y}/\sumMatrix{A}$
where $\sumCol{A}{x}\defeq \sum_{y=1}^{\N{}} A[y][x]$.

Let $Z$ and $W$ be two random variables. 
The \textit{Shannon entropy}, $\Hp{Z}$, of $Z$~\cite{Shannon} 
evaluates  the information carried by $Z$ itself. 
In the general case, it is formally defined as\begin{equation}
\Hp{Z} \defeq - \sum_{z \in \mc{Z}} \genP[Z]{z} \log_2 \genP[Z]{z}
\label{Entropy:def}
\end{equation}
where $\genP[Z]{z}$ is the probability for $Z$ to get the value $z$ and $\mc{Z}$ is the set of all the possible values for it.  
Without making any assumption on $\genP[Z]{z}$,  
$\Hp{Z}$ can be proved to belong to the closed interval $[0,\log_2 |\mc{Z}|]$. 
It is worth to underline that, since $0$ is not 
included in the domain of the logarithmic function, $\Hp{Z}$ is 
well-defined if and only if $\genP[Z]{z}>0$ for all $z \in \mc{Z}$.
Moreover, the following proposition holds.
\begin{prop}\label{prop:Hgt0}
If $\Hp{Z}$ is well-defined 
and $|\mc{Z}|>1$, 
$\Hp{Z}>0$. 
\end{prop}
\begin{proof}
If $\Hp{Z}$ is well-defined, then 
$\genP[Z]{z}>0$ for all $z \in \mc{Z}$. 
Thus, $\genP[Z]{z} \in (0,1]$, $\log_2 \genP[Z]{z}$ is 
non-positive, and so 
$\genP[Z]{z} * \log_2 \genP[Z]{z}$ is. 
It follows that $\Hp{Z}$ equals $0$ if and only if 
all its terms -- i.e., $\genP[Z]{z} * \log_2 \genP[Z]{z}$ -- equal 
$0$, but this exclusively happens when $\genP[Z]{z} = 1$.
However,  by definition of probability function, 
$\sum_{z \in \mc{Z}}\genP[Z]{z} = 1$. 
We can conclude that either $|\mc{Z}|=1$, which contradicts 
the proposition's hypothesis, 
or $0<\genP[Z]{z} < 1$ for all $z \in \mc{Z}$ and $\Hp{Z}>0$.\qed
\end{proof}

The \emph{conditional entropy of $W$ given $Z$}~\cite{Shannon}
measures the quantity of information in $W$ when an insight of $Z$ is available and it is defined as
\begin{equation}\label{condH:def}
\Hp{W/Z} \defeq -\sum_{z \in \mc{Z}}\sum_{w \in \mc{W}} \genP[ZW]{w,z}
\log_2 \frac{\genP[ZW]{w,z}}{\genP[Z]{z}}
\end{equation}
where $\genP[ZW]{w,z}$ is the \emph{joint probability} for both $Z$ and $W$ to 
get the values $z$ and $w$  at the same time and $\mc{W}$ is the set of 
all the possible values for $w$.

The \emph{mutual information} $\MI(Z,W)$ 
measures how far are $Z$ and $W$ from 
being independent, i.e., it gauges how much the values that they assume 
are related still being potentially different.
$\MI(Z,W)$ is formally defined as:
\begin{equation}\label{MI:def}
\MI{}(Z,W) \defeq  \sum_{z \in \mc{Z}}\sum_{w \in \mc{W}}
\genP[ZW]{w,z}\log\frac{\genP[ZW]{w,z}}{\genP[Z]{z}*\genP[W]{w}}
\end{equation}
and it is easy to prove that 
\begin{equation}
\MI{}(Z,W) = \Hp{Z}+\Hp{W}-\Hp{ZW} =\MI{}(W,Z) \geq 0
\label{MXY:eq}
\end{equation}

Given the probability distributions $P_{X_A}=\{\Px{x}\}_{x}$,  $P_{Y_A}=\{\Py{y}\}_{y}$, and  $P_{X_AY_A}=\{\Pxy{y,x}\}_{x,y}$, 
the entropy values for the so-called \emph{marginal random variables} $X_A$ and $Y_A$ -- i.e., $\Hp{X_A}$ and $\Hp{Y_A}$, respectively -- and for the random variable $X_AY_A$ 
-- i.e., $\Hp{X_AY_A}$ --  
can be computed as shown by Eq.~\ref{Entropy:def}. As a consequence, 
the mutual information between $X_A$ and $Y_A$ 
can be evaluated too.
All these quantities are completely determined by the agreement matrix 
because $\Px{x}$, 
$\Py{y}$, and $\Pxy{y,y}$ exclusively depend on $A$ itself.
Moreover, it can be proved that
 $\Hp{X_A}=\Hp{Y_{\trans{A}}}$, $\Hp{Y_A}=\Hp{X_{\trans{A}}}$, and 
$\Hp{X_AY_A}=\Hp{X_{\trans{A}}Y_{\trans{A}}}$
where $\trans{A}$ denotes 
the transposed matrix of $A$, i.e., $\trans{A}[x][y]=A[y][x]$ for 
all rows $y$ and for all columns $x$ in $A$.

The \emph{information agreement} (\IA{}) of $A$~\cite{IA2020} was 
introduced to 
gauge the agreement between the two raters $\mf{X}$ and $\mf{Y}$ on the 
data set $\mathcal{D}$ by considering $A$.
It is formally defined as follows:
\begin{equation}\label{IA:def}
\IA{}(A)\defeq\frac{\MI{}(X_A,Y_A)}{\min\{\Hp{X_A},\Hp{Y_A}\}}.
\end{equation}
It is known that the information agreement 
is not well-defined for all the agreement matrices $A$. 
However, whenever $\IA{}(A)$ is defined, its value belongs to 
the interval $[0,1]$.

%% file: extending_IA.tex

\section{Extending \IA{}}

In its original form, the information agreement is not well-defined for 
all the possible agreement matrices $A$.
In particular, since \IA{} is the ratio between $\MI{}(X_A,Y_A)$ and 
$\min\{\Hp{X_A},\Hp{Y_A}\}$ (see Eq.~\ref{IA:def}) and 
$\MI{}(X_A,Y_A)$ equals the sums and subtractions 
of entropies (see Eq.~\ref{MXY:eq}), 
\IA{} is not defined under two circumstances: when at least one entropy 
among $\Hp{X_A}$, $\Hp{Y_A}$, and $\Hp{X_AY_A}$ is not defined and 
when the minimum among $\Hp{X_A}$ and $\Hp{Y_A}$ is $0$.
According to what we noticed in Section~\ref{sec:basics}, 
the former case exclusively occurs when there exist $x,y \in [1,\N{}]$ such that either $\genP[X_A]{x}=0$, $\genP[Y_A]{y}=0$, or 
$\genP[X_AY_A]{y,x}=0$. However, by definition of $\genP[X_A]{x}$, $\genP[Y_A]{y}$, and  $\genP[X_AY_A]{y,x}$, this is equivalent to 
the existence of a value in $A$ that equals $0$.  
As far as the latter case may concern, if both 
$\Hp{X_A}$ and $\Hp{Y_A}$ are well-defined, then 
both $\Hp{X_A}$ and $\Hp{Y_A}$ are greater than $0$ by 
Prop.~\ref{prop:Hgt0} because $\N{}>$ by assumption. 
It follows that $\IA{}(A)$ is well-defined if 
and only if all the values in $A$ are greater than $0$.

Since the logarithmic function is defined and continuous in the interval 
$(0,+\infty)$, one possible solution to overcome the inability of 
computing \IA{} on an agreement matrix $A$ containing some $0$  
is to build a new symbolic agreement matrix $\eps{A}$ 
that replaces all the occurrences of $0$ in $A$ with a real variable 
$\epsilon$.  
The matrix $\eps{A}$ is the \emph{$0$-freed matrix} and it is 
formally defined as follows:
\begin{equation*}
\eps{A}[y][x]\defeq \left\{\begin{array}{ll}
A[y][x]&\textrm{if $A[y][x]\neq 0$}\\
\epsilon&\textrm{otherwise}
\end{array}\right.
\end{equation*} 
where $\epsilon$ is a real variable assuming values in the open interval $(0,+\infty)$.

Because of their definitions, it is easy to see that 
$\Pxy[\eps{A}]{y,x}$,  
$\Px[\eps{A}]{x}$, and $\Py[\eps{A}]{y}$ 
belong to the real interval $(0,1)$ for all $x,y \in [1,n]$ and 
for all $\epsilon \in (0,+\infty)$. It follows that $\Hp{X_{\eps{A}}Y_{\eps{A}}}$, 
$\Hp{X_{\eps{A}}}$, and $\Hp{Y_{\eps{A}}}$ are  well-defined for 
any positive value of $\epsilon$ and so $\IA{}(\eps{A})$ is.
Thus, the limit 
for $\IA{}(\eps{A})$ as $\epsilon$ tends to $0$ from the right 
may be a reasonable estimation for $\IA{}(A)$.

It is worth to underline that, while $\IA{}(A)$, when defined, is 
a value, $\IA{}(\eps{A})$ is a function on $\epsilon$ whose domain 
is open real interval $(0,+\infty)$ and, because 
of this, its limit as $\epsilon$ tends to $0$ from the right 
may not exist.
However, if this limit does exist, then it will be the 
extension-by-continuity of \IA{} over the matrix $A$. 
This limit is the \emph{Information Agreement extension by Continuity} $\IAs{}$ and is formally defined as follows:
\begin{equation}\label{IAs:def}
\IAs{}(A) \defeq \lim_{\epsilon \rightarrow 0^+}\IA{}(\eps{A})
\end{equation}

In the following part of this section, we will prove that 
$\IAs{}(A)$ always exists and we show how to compute it.
This achievement will be eased by the following proposition.
\begin{lem}\label{lem:prob}
Let ${\gAM{}}$ be an $n\times n$-agreement matrix. For all $v,w \in [1,n]$,
$\Px[{\gAM{}}]{v}=\Py[\trans{\gAM{}}]{v}$ and 
$\Pxy[{\gAM{}}]{w,v}=\Pxy[\trans{\gAM{}}]{w,v}$.
\end{lem}
\begin{proof}
By the definitions of $\sumCol{\gAM{}}{x}$ and $\sumMatrix{\gAM{}}$, $\sumCol{\gAM{}}{x}\defeq \sum_{y=1}^{n} {\gAM{}}[y][x]$ and 
$\sumMatrix{\gAM{}}\defeq \sum_{x=1}^{n} \sum_{y=1}^{n} {\gAM{}}[y][x]$. 
So, because of the definition of $\trans{\gAM{}}$, 
\begin{align*}
\sumCol{\gAM{}}{x}&= \sum_{y=1}^{n} {\gAM{}}[y][x]= \sum_{y=1}^{n} \trans{\gAM{}}[x][y] = \sumRow{\trans{\gAM{}}}{x}
\end{align*}
and, analogously,
\begin{align*}
\sumMatrix{\gAM{}}&= \sum_{x=1}^{n} \sum_{y=1}^{n} {\gAM{}}[y][x] = \sum_{x=1}^{n} \sum_{y=1}^{n} \trans{\gAM{}}[x][y] = \sumMatrix{\trans{\gAM{}}}.
\end{align*}
Since $\Px[{\gAM{}}]{x}\defeq\sumCol{\gAM{}}{x}/\sumMatrix{\gAM{}}$ and 
$\Py[\trans{\gAM{}}]{x} \defeq\sumRow{\trans{\gAM{}}}{x}/\sumMatrix{\trans{\gAM{}}}$ by definition, it follows that $\Px[{\gAM{}}]{x}=\Py[\trans{\gAM{}}]{x}$. 

Moreover, ${\gAM{}}[y][x]=\trans{\gAM{}}[x][y]$ for all $x,y \in [1,n]$ by definition 
of transposed matrix. Hence, $\Pxy[{\gAM{}}]{y,x}=\Py[\trans{\gAM{}}]{x,y}$ for all $x,y \in [1,n]$, because $\Pxy[{\gAM{}}]{x}\defeq\sumCol{\gAM{}}{x}/\sumMatrix{\gAM{}}$ 
and $\Pxy[\trans{\gAM{}}]{x,y}\defeq\Pxy[\trans{\gAM{}}]{x,y}/\sumMatrix{\trans{\gAM{}}}$.
\end{proof}

Thanks to Lemma~\ref{lem:prob} which unravels the relation between 
the probability function associated to an agreement matrix ${\gAM{}}$ 
and that of $\trans{\gAM{}}$, we can easily prove the following 
proposition about the entropy functions.

\begin{lem}\label{lem:transpose}
Let ${\gAM{}}$ be an $n\times n$-agreement matrix such that ${\gAM{}}[y][x]>0$ for all 
rows $y$ and for all columns $x$ in ${\gAM{}}$.
%
The following equalities hold:
\begin{enumerate}
\item 
$\Hp{X_{\gAM{}}}=\Hp{Y_{\trans{\gAM{}}}}$;

\item 
$\Hp{Y_{\gAM{}}}=\Hp{X_{\trans{\gAM{}}}}$;

\item 
$\Hp{X_{\gAM{}}Y_{\gAM{}}}=\Hp{X_{\trans{\gAM{}}}Y_{\trans{\gAM{}}}}$.
\end{enumerate}
\end{lem}
\begin{proof}
Let us prove the claim, point by point.
\begin{enumerate}
\item\label{prop:transpose:point1} 
By Eq.~\ref{Entropy:def} and by Lemma~\ref{lem:prob}, it is immediate to see that 
\begin{align*}
\Hp{X_{\gAM{}}}&=-\sum_{x=1}^{n} \Px[{\gAM{}}]{x} \log_2 \Px[{\gAM{}}]{x}\\
&=-\sum_{x=1}^{n} \Py[\trans{\gAM{}}]{x} \log_2 \Py[\trans{\gAM{}}]{x}=\Hp{Y_{\trans{\gAM{}}}}
\end{align*}
\item Let $C$ be the matrix $\trans{\gAM{}}$. 
So, $\Hp{X_{C}}=\Hp{Y_{\trans{C}}}$ by Point~\ref{prop:transpose:point1}.
However, it is easy to see that $\trans{C}=\trans{\left(\trans{\gAM{}}\right)}={\gAM{}}$ and, thus, that  $\Hp{Y_{\gAM{}}}=\Hp{Y_{\trans{C}}}=\Hp{X_{C}}=\Hp{X_{\trans{\gAM{}}}}$.
\item Because of Lemma~\ref{lem:prob}, we know that 
$\Pxy{y,x}=\Pxy[\trans{\gAM{}}]{x,y}$ for any $x,y \in [1,n]$. 
It follows that 
\begin{align*}
\Hp{X_AY_A}&=-\sum_{x=1}^{n} \sum_{y=1}^{n} \Pxy{y,x} \log_2 \Pxy{y,x}\\
&=-\sum_{x=1}^{n} \sum_{y=1}^{n} \Pxy[\trans{\gAM{}}]{x,y} \log_2 \Pxy[\trans{\gAM{}}]{x,y}=\Hp{X_AX_Y}
\end{align*}
\end{enumerate}
This ends the proof of the claim.
\end{proof}
From Lemma~\ref{lem:transpose} trivially follows the following 
claim.
\begin{prop}\label{prop:trans:IA}
Let ${\gAM{}}$ be an $n\times n$-agreement matrix such that
${\gAM{}}[y][x]>0$ for all 
rows $y$ and for all columns $x$ in ${\gAM{}}$. 
It holds that:
\begin{itemize}
	\item $\MI{}(X_{\trans{\gAM{}}},Y_{\trans{\gAM{}}})=\MI{}(X_{\gAM{}},Y_{\gAM{}})$;
	\item $\min\{\Hp{X_{\gAM{}}},\Hp{Y_{\gAM{}}}\}=\min\{\Hp{X_{\trans{\gAM{}}}}, \Hp{Y_{\trans{\gAM{}}}}\}$;
	\item $\IA{}({\gAM{}})=\IA{}(\trans{\gAM{}})$.
\end{itemize}
\end{prop}
\begin{proof}
Due of Lemma~\ref{lem:transpose} and  
Eq.~\ref{MXY:eq}, it is easy to see that, for any $n\times n$-matrix ${\gAM{}}$ whose values are all positive, both
$\MI{}(X_{\trans{\gAM{}}},Y_{\trans{\gAM{}}})$ equals $\MI{}(X_{\gAM{}},Y_{\gAM{}})$  and 
$\min\{\Hp{X_{\gAM{}}},$ $\Hp{Y_{\gAM{}}}\}$ equals $\min\{\Hp{X_{\trans{\gAM{}}}},$ $\Hp{Y_{\trans{\gAM{}}}}\}$.
Moreover, both 
$\Hp{X_{\gAM{}}}$ and $\Hp{Y_{\gAM{}}}$ are well-defined because 
${\gAM{}}[y][x]>0$ for all 
rows $y$ and for all columns $x$ in ${\gAM{}}$ by hypothesis.
Hence, since $n>1$ by assumption  both 
$\Hp{X_{\gAM{}}}$ and $\Hp{Y_{\gAM{}}}$ are greater than $0$  
by Prop.~\ref{prop:Hgt0} and so 
$\min\{\Hp{X_{\gAM{}}},$ $\Hp{Y_{\gAM{}}}\}$  is.
Because of the definition of \IA{} (see Eq.~\ref{IA:def}), 
the claim directly follows.
\end{proof}

When the function $\IA{}(\eps{A})$ is studied,
$\Hp{X_{\eps{A}}}$ can be assumed to be smaller than or equal to 
$\Hp{Y_{\eps{A}}}$ without any loss of generality. 
Indeed, if this is not the case --i.e., if $\Hp{Y_{\eps{A}}}<\Hp{X_{\eps{A}}}$--, 
the function $\IA{}(\transeps{A})$, which equals $\IA{}(\eps{A})$ by 
Prop.~\ref{prop:trans:IA}, 
can be considered in place of $\IA{}(\eps{A})$ itself, and, by  
Lemma~\ref{lem:transpose}, 
we know that $\Hp{X_{\transeps{A}}}=\Hp{Y_{\eps{A}}}<\Hp{X_{\eps{A}}}=\Hp{Y_{\transeps{A}}}$ will hold.

If $\Hp{X_{\eps{A}}}\leq \Hp{Y_{\eps{A}}}$, then  
$\Hp{X_{\eps{A}}}=\min\{\Hp{X_{\eps{A}}},\Hp{Y_{\eps{A}}}\}$. Thus,  
by Eq.~\ref{MXY:eq} and~\ref{IA:def}, 
$\IA{}(\eps{A}) = 1 + (\Hp{Y_{\eps{A}}}-\Hp{X_{\eps{A}}Y_{\eps{A}}})/\Hp{X_{\eps{A}}}$ 
and, because of continuity of $+$ on $\mathbb{R}\times\mathbb{R}$, 
if $\IAs{}(A)$ exists, then 
\begin{equation}\label{eq:limit}
\IAs{}(A)=1+\lim_{\epsilon \rightarrow 0^+}\frac{\Hp{Y_{\eps{A}}}-\Hp{X_{\eps{A}}Y_{\eps{A}}}}{\Hp{X_{\eps{A}}}}.
\end{equation}

In order to evaluate above formula, let us first introduce a 
function to restrict the domain of a generic random variable to 
those values that have probability greater than $0$.
\begin{definition}
Let $Z$ be a random variable getting values from $\mc{Z}$ and 
such that $\genP[Z]{z}$ is the probability for $Z$ to have the value 
$z \in \mc{Z}$. 

The \emph{refined random variable of $Z$}, denoted by 
$\nnrv{Z}$, is a random variable getting values from the set 
$\nnrv{\mc{Z}} \defeq \{z\in \mc{Z} |\ \genP[Z]{z} >0 \}$ 
which contains all the values in $\mc{Z}$ that have 
non-null probability with respect to $\genP[Z]{\cdot}$.
\end{definition}
It is worth to notice that $\genP[Z]{z}=\genP[\nnrv{Z}]{z}$ for any 
value in the domain of $\nnrv{Z}$.

The following proposition relates the entropy functions associated to 
$X_{\eps{A}}$, $Y_{\eps{A}}$, and $X_{\eps{A}}Y_{\eps{A}}$ to those 
associated to $\nnrv{X_{A}}$, $\nnrv{Y_{A}}$, and $\nnrv{X_{A}Y_{A}}$, 
respectively.
\begin{prop}\label{prop:nnrv}
Let ${\gAM{}}$ an agreement matrix.
The following equation holds:
\begin{itemize}
\item $\lim_{\epsilon \rightarrow 0^+}\Hp{X_{\eps{\gAM{}}}}=\Hp{\nnrv{X_{\gAM{}}}}$
\item $\lim_{\epsilon \rightarrow 0^+}\Hp{Y_{\eps{\gAM{}}}}=\Hp{\nnrv{Y_{\gAM{}}}}$
\item $\lim_{\epsilon \rightarrow 0^+}\Hp{X_{\eps{\gAM{}}}Y_{\eps{\gAM{}}}}=\Hp{\nnrv{X_{\gAM{}}Y_{\gAM{}}}}$
\end{itemize}
\end{prop}
\begin{proof}
Let us focus on the first equation: the correctness of the other two equations can be proved in an analogous way.
By definition, 
\begin{equation*}
\Hp{X_{\eps{\gAM{}}}}\defeq-\sum_{x=1}^n \Px[\eps{\gAM{}}]{x}*\log_2 \Px[\eps{\gAM{}}]{x}
\end{equation*}
Thus, by the continuity of both $+$ and $*$ on $\mathbb{R}\times\mathbb{R}$,
\begin{equation*}
\lim_{\epsilon \rightarrow 0^+}\Hp{X_{\eps{\gAM{}}}}=-\sum_{x=1}^n \lim_{\epsilon \rightarrow 0^+}\left(\Px[\eps{\gAM{}}]{x}*\log_2 \Px[\eps{\gAM{}}]{x}\right).
\end{equation*}
However, we know that $\Px[\eps{\gAM{}}]{x} \defeq \sumCol{\eps{\gAM{}}}{x}/\sumMatrix{\eps{\gAM{}}}$ and that 
$\sumCol{\eps{\gAM{}}}{x}\defeq \sum_{x=1}^{n} \eps{\gAM{}}[y][x]$ and 
$\sumMatrix{\eps{\gAM{}}}\defeq \sum_{x=1}^{n} \sum_{y=1}^{n} \eps{\gAM{}}[y][x]$.
Since all the values in ${\gAM{}}$ are non-negative, all the non-symbolic 
values in $\eps{\gAM{}}$ are positive by construction.
Thus, because of the continuity of $+$ on $\mathbb{R}\times\mathbb{R}$, 
$\lim_{\epsilon \rightarrow 0^+}\sumMatrix{\eps{\gAM{}}} = \sumMatrix{\gAM{}}$ and, 
since we assumed that every agreement matrix contains at least one 
non-null value, $\sumMatrix{\gAM{}}>0$.
Analogously, 
$\lim_{\epsilon \rightarrow 0^+}\sumCol{\eps{\gAM{}}}{x} = \sumCol{\gAM{}}{x}$ and 
$\sumCol{\gAM{}}{x}\geq 0$.
So, due to the continuity of $/$ on $\mathbb{R}\times\mathbb{R}_{>0}$,
$\lim_{\epsilon \rightarrow 0^+}\Px[\eps{\gAM{}}]{x} = 0$ 
if and only if $\sumCol{\eps{\gAM{}}}{x} = n*\epsilon$ or, equivalently,
if and only if $\sumCol{\gAM{}}{x} = 0$. 

So, every time  $\sumCol{\gAM{}}{x} > 0$, $\Px[B]{x} >0$ by definition 
of $\Px[B]{x}$ and 
\begin{equation*}
\lim_{\epsilon \rightarrow 0^+}\left(\Px[\eps{\gAM{}}]{x}*\log_2 \Px[\eps{\gAM{}}]{x}\right) = \Px[B]{x}*\log_2 \Px[B]{x} < 0
\end{equation*}
because of the continuity of both $*$ 
on $\mathbb{R}\times\mathbb{R}$ and $\log$ on $\mathbb{R}\times\mathbb{R}_{>0}$.
If instead $\sumCol{\gAM{}}{x} = 0$, it easy to prove, by using the de 
l'H\^opital's rule, that the limit for $\Px[\eps{\gAM{}}]{x}*\log_2 \Px[\eps{\gAM{}}]{x}$ as $\epsilon$ tends to $0$ from the right is $0$.

It follows that
\begin{align*}
\lim_{\epsilon \rightarrow 0^+}\Hp{X_{\eps{\gAM{}}}}&=-\sum_{x=1}^n \lim_{\epsilon \rightarrow 0^+}\left(\Px[\eps{\gAM{}}]{x}*\log_2 \Px[\eps{\gAM{}}]{x}\right)\\
&=-\left(\sum_{x \in \nnrv{[1,n]}} \Px[B]{x}*\log_2 \Px[B]{x}\right)-\sum_{x \in [1,n]\setminus \nnrv{[1,n]}} 0,
\end{align*}
where $\nnrv{[1,n]}$ is the set $\{x \in [1,n]\ |\ \Px[B]{x}>0\}$, 
and, by definition of $\nnrv{X_B}$, 
\begin{align*}
\lim_{\epsilon \rightarrow 0^+}\Hp{X_{\eps{\gAM{}}}}=\Hp{\nnrv{X_B}}.
\end{align*}
This concludes the proof for the first equation in the claim. The 
proof of the correctness of the remaining equations is analogous. 
\end{proof}

Thanks to the continuity of both $-$ on $\mathbb{R}\times\mathbb{R}$ and 
$/$ on $\mathbb{R}\times\mathbb{R}_{>0}$, 
Prop.~\ref{prop:nnrv} proves that, whenever $\Hp{\nnrv{X_{A}}}$ is greater than $0$ and smaller than $\Hp{\nnrv{Y_{A}}}$, 
$\IAs{}(A)$ exists and it can be easily computed as 
$\IAs{}(A)=1+({\Hp{\nnrv{Y_{A}}}-\Hp{\nnrv{X_{A}}\nnrv{Y_{A}}}})/{\Hp{\nnrv{X_{A}}}}$.
%
This statement is summarized in the following theorem.

\begin{thm}\label{theo:easy}
Let ${\gAM{}}$ be an $n\times n$-agreement matrix. 
If 
$0 < \Hp{\nnrv{X_{\gAM{}}}}\leq \Hp{\nnrv{Y_{\gAM{}}}}$,
then $\IAs{}({\gAM{}})$ exists and it equals:
\begin{equation*} \IAs{}({\gAM{}})=1+\frac{\Hp{\nnrv{Y_{\gAM{}}}}-\Hp{\nnrv{X_{\gAM{}}}\nnrv{Y_{\gAM{}}}}}{\Hp{\nnrv{X_{\gAM{}}}}}.
\end{equation*}
\end{thm}

Intriguingly, Lemma~\ref{lem:transpose} can be extended to deal with 
refined random variables.

\begin{lem}~\label{lem:transpose_nonnull}
	Let ${\gAM{}}$ be an agreement matrix.
	%
	The following equalities hold:
	\begin{enumerate}
		\item 
		$\Hp{\nnrv{X_{\gAM{}}}}=\Hp{\nnrv{Y_{\trans{\gAM{}}}}}$;
		
		\item 
		$\Hp{\nnrv{Y_{\gAM{}}}}=\Hp{\nnrv{X_{\trans{\gAM{}}}}}$;
		
		\item 
		$\Hp{\nnrv{X_{\gAM{}}Y_{\gAM{}}}}=\Hp{\nnrv{X_{\trans{\gAM{}}}Y_{\trans{\gAM{}}}}}$.
	\end{enumerate}
\end{lem}
\begin{proof}
	By Prop.~\ref{prop:nnrv} 
	$\Hp{\nnrv{X_{\gAM{}}}}$, $\Hp{\nnrv{Y_{\gAM{}}}}$, and 
	$\Hp{\nnrv{X_{\gAM{}}Y_{\gAM{}}}}$ equal the limits 
	as $\epsilon$ tends to $0$ from the right for 
	$\Hp{{X_{\eps{\gAM{}}}}}$, $\Hp{{Y_{\eps{\gAM{}}}}}$, and 
	$\Hp{{X_{\eps{\gAM{}}}Y_{\eps{\gAM{}}}}}$, respectively.
	
	However, by Lemma~\ref{lem:transpose}, 
	$\Hp{{X_{\eps{\gAM{}}}}}=\Hp{{Y_{\trans{\eps{\gAM{}}}}}}$, $\Hp{{Y_{\eps{\gAM{}}}}}=\Hp{{Y_{\trans{\eps{\gAM{}}}}}}$, and 
	$\Hp{{X_{\eps{\gAM{}}}Y_{\eps{\gAM{}}}}}=\Hp{{X_{\trans{\eps{\gAM{}}}}}Y_{\trans{\eps{\gAM{}}}}}$ for any $\epsilon>0$. 
	
	By Prop.~\ref{prop:nnrv}, $\Hp{\nnrv{X_{\trans{\gAM{}}}}}$, $\Hp{\nnrv{Y_{\trans{\gAM{}}}}}$, and 
	$\Hp{\nnrv{X_{\trans{\gAM{}}}Y_{\trans{\gAM{}}}}}$ equal the limits 
	as $\epsilon$ tends to $0$ from the right for 
	$\Hp{{Y_{\trans{\eps{\gAM{}}}}}}$, $\Hp{{Y_{\trans{\eps{\gAM{}}}}}}$, and 
	$\Hp{{X_{\trans{\eps{\gAM{}}}}}Y_{\trans{\eps{\gAM{}}}}}$, respectively.
	This concludes the proof of the claim.
\end{proof}

Thanks to Lemma~\ref{lem:transpose_nonnull}, it is easy to see 
that   $\IAs{}(A)=\IAs{}(\trans{A})$. Moreover, if $\Hp{\nnrv{X_{A}}}> \Hp{\nnrv{Y_{A}}}$, 
then $\Hp{\nnrv{X_{\trans{A}}}} < \Hp{\nnrv{Y_{\trans{A}}}}$ by 
the same lemma. 
Hence, Theorem~\ref{theo:easy} deals with
all the agreement matrices $A$ for which both 
$\Hp{\nnrv{X_{A}}}$ and $\Hp{\nnrv{Y_{A}}}$ are greater than $0$.

\begin{table}[!ht]
	\begin{subtable}[b]{0.47\linewidth}
		\scalebox{0.87}{\begin{minipage}[c]{\textwidth}
				\[
				A^*=\begin{pmatrix}
				a_1&0&\ldots&\ldots&\ldots&\ldots&0\\
				\vdots&\vdots&\ddots&&&&\vdots\\
				a_m&\vdots&&\ddots&&&\vdots\\
				0&\vdots&&&\ddots&&\vdots\\
				\vdots&\vdots&&&&\ddots&\vdots\\
				0&0&\ldots&\ldots&\ldots&\ldots&0\\
				\end{pmatrix}
				\]
		\end{minipage}}
		\caption{An agreement matrix such that 
			$\Hp{\nnrv{X_{A^*}}}=0$ and 
			$\Hp{\nnrv{X_{A^*}}}<\Hp{\nnrv{Y_{A^*}}}$. This 
			matrix does not satisfy the hypothesis of Theorem~\ref{theo:easy}.}\label{table:single_column_A}
	\end{subtable}
	\hfill
	\begin{subtable}[b]{0.47\linewidth}
		\scalebox{0.87}{\begin{minipage}[c]{\textwidth}
				\[
				\hspace{0.5cm}
				\eps{A}^*=\begin{pmatrix}
				a_1&\epsilon&\ldots&\ldots&\ldots&\ldots&\epsilon\\
				\vdots&\vdots&\ddots&&&&\vdots\\
				a_m&\vdots&&\ddots&&&\vdots\\
				\epsilon&\vdots&&&\ddots&&\vdots\\
				\vdots&\vdots&&&&\ddots&\vdots\\
				\epsilon&\epsilon&\ldots&\ldots&\ldots&\ldots&\epsilon\\
				\end{pmatrix}
				\]
		\end{minipage}}
		\caption{This matrix is obtained from the agreement matrix $A^*$ 
			reported in Table~\ref{table:single_column_A} by replacing all 
			the $0$s by the real variable $\epsilon$.}\label{table:A_epsilon}
	\end{subtable}
\end{table}

In order to complete our analysis, it is worth to understand under which 
conditions $\Hp{\nnrv{X_{A}}}$ equals $0$. By definition of entropy,
\begin{equation*}
\Hp{\nnrv{X_{A}}}\defeq - \sum_{x \in \nnrv{[1,n]}} \Px[\nnrv{X_{A}}]{x}*\log_2 \Px[\nnrv{X_{A}}]{x}
\end{equation*}
where $\nnrv{[1,n]}\defeq \{x \in [1,n]\ |\ \Px[X_{A}]{x}>0\}$. 
Since $\sum_{x \in \nnrv{[1,n]}} \Px[\nnrv{X_{A}}]{x}=1$ by definition 
of probability, $\Hp{\nnrv{X_{A}}}=0$ if and only if $\nnrv{[1,n]}$ 
contains exclusively one column $\nnrv{x}$ whose probability is $1$, 
i.e., $\Px[A]{\nnrv{x}}=\genP[\nnrv{X_{A}}]{\nnrv{x}}=1$. 
Because of the definition of $\Px[A]{x}$, this means that $\nnrv{x}$ is the only column in $A$ whose values are not all
$0$ or, equivalently, that $\nnrv{x}$ is the only \emph{non-null} 
column in $A$. 
Thus, to prove the existence of $\IAs{}(A)$ for any agreement matrix 
$A$, we need to solve Eq.~\ref{eq:limit} when $A$ is 
a generic agreement matrix having 
exclusively one 
non-null column or row.
As already observed above, the two cases are symmetrical and 
we can focus on one of the two cases. Let us consider an agreement 
matrix 
having exclusively one non-null column and $m$ non-null rows.
For the 
sake of simplicity and without any loss in generality, 
we will impose that the values different from $0$ are those 
contained in the column $1$ and in the first $m$ rows 
as in the matrix $A^*$ 
depicted by Table~\ref{table:single_column_A}. 
This assumption does not weaken the generality of the 
considered case because the entropy functions and, consequently, 
the information agreement do not take into account the position 
of classification events in the agreement matrix, but exclusively their 
probabilities.

Table~\ref{table:A_epsilon} reports 
the $0$-freed matrix of $A^*$.
Since  
$\sumCol{\eps{A}^*}{x}\defeq \sum_{y=1}^{n} \eps{A}^*[y][x]$ 
and 
$\sumRow{\eps{A}^*}{y}\defeq \sum_{x=1}^{n} \eps{A}^*[y][x]$ 
by definition, it is easy to see that 
\begin{equation}\label{eq:defA*Col}
\sumCol{\eps{A}^*}{x} = \left\{\begin{array}{ll}
(n-m)*\epsilon + \sum_{y=1}^{m} a_y&\textrm{if $x=1$}\\
n*\epsilon& \textrm{otherwise}\\
\end{array}\right.,
\end{equation}
and, analogously,
\begin{equation}\label{eq:defA*Row}
\sumRow{\eps{A}^*}{y} = \left\{\begin{array}{ll}
(n-1)*\epsilon + a_y&\textrm{if $y\in [1,m]$}\\
n*\epsilon& \textrm{otherwise}\\
\end{array}\right.
\end{equation}
As far as $\sumMatrix{\eps{A}^*}$ may concern, it is easy to see 
that $\sumMatrix{\eps{A}^*}=(n^2-m)*\epsilon+\sum_{y=1}^{m} a_y$.

The following preparatory lemma is meant to syntactically simplify 
Eq.~\ref{eq:limit}.

\begin{lem}\label{lemma:simplify_H}
	Let $Z$ be a random variable that assumes values in $\mc{Z}$ and 
	let $\genP[Z]{z}$ be the probability for $Z$ to get the value $z$.
	
	If  $\genP[Z]{z}=f(z)/c$, where $c \in \mathbb{R}\setminus\{0\}$ is 
	a constant value and $f:\mc{Z}\rightarrow \mathbb{R}$ is 
	function such that $\sum_{z \in \mc{Z}} f(z)=c$, then the following 
	equation holds:
	\begin{equation}
	\Hp{Z} = \log_2 c - \frac{1}{c}*\sum_{z \in \mc{Z}} f(z) * \log_2 f(z).
	\end{equation}
\end{lem}
\begin{proof}
	Since $\Hp{Z} \defeq - \sum_{z \in \mc{Z}}  \genP[Z]{z} * \log_2 \genP[Z]{z}$ by definition, it holds that 
	\begin{align*}
	\Hp{Z} &= - \sum_{z \in \mc{Z}}  \genP[Z]{z} * \log_2 \genP[Z]{z}\\
	&= - \sum_{z \in \mc{Z}} \frac{f(z)}{c} * \log_2 \frac{f(c)}{c}\\
	&=-\frac{1}{c}\sum_{z \in \mc{Z}} \left(f(z) * \left(\log_2 f(z) - \log_2 c\right)\right)\\
	&=\frac{1}{c}\left(\left(\sum_{z \in \mc{Z}} f(z) *\log_2 c\right) - \left(\sum_{z \in \mc{Z}} f(z) * \log_2 f(z)\right)\right)\\
	&=\frac{1}{c}\left(\left(\sum_{z \in \mc{Z}} f(z)\right)*\log_2 c - 
	\left(\sum_{z \in \mc{Z}} f(z) * \log_2 f(z)\right) \right)\\
	\end{align*}
	However, $\sum_{z \in \mc{Z}}^n f(z) = c$ by hypothesis and, then, 
	\begin{align*}
	\Hp{Z}
	&=\log_2 c -  \frac{1}{c}*\sum_{z \in \mc{Z}} f(z) * \log_2 f(z)
	\end{align*}
	This concludes the proof of the claim. 
\end{proof}

It is easy to see that if ${\gAM{}}$ is an $n \times n$-agreement matrix 
(potentially, also $0$-freed), then the variable $X_{\gAM{}}$, 
$Y_{\gAM{}}$, $X_{\gAM{}}Y_{\gAM{}}$ satisfy the conditions of 
Lemma~\ref{lemma:simplify_H} and the equations 
\begin{equation}\label{eq:simple1}
\Hp{X_{\gAM{}}} = \log_2 \sumMatrix{\gAM{}} - \frac{1}{\sumMatrix{\gAM{}}}*\sum_{x=1}^n \hs[2]{\sumCol{\gAM{}}{x}},
\end{equation}
\begin{equation}\label{eq:simple2}
\Hp{Y_{\gAM{}}} = \log_2 \sumMatrix{\gAM{}} - \frac{1}{\sumMatrix{\gAM{}}}*\sum_{y=1}^n \hs[2]{\sumRow{\gAM{}}{y}},
\end{equation}
and 
\begin{equation}\label{eq:simple3}
\Hp{X_BY_B} = \log_2 \sumMatrix{\gAM{}} - \frac{1}{\sumMatrix{\gAM{}}}*\sum_{y=1}^n \sum_{x=1}^n \hs[2]{B[y][x]},
\end{equation}
hold. 

Let us introduce the shortcuts 
$\THr{\gAM{}}\defeq (\ln{2})* \sumMatrix{\gAM{}}*\left(\Hp{Y_{\gAM{}}}-\Hp{X_{\gAM{}}Y_{\gAM{}}}\right)$, 
$\BHr{\gAM{}}\defeq (\ln{2})*\sumMatrix{\gAM{}}*\Hp{X_{\gAM{}}}$,
and $\Hr{\gAM{}}\defeq \THr{\gAM{}}/\BHr{\gAM{}}$. It is worth to 
notice that, for all $\epsilon >0$,    $\Hr{\eps{A}^*}=({\Hp{Y_{\eps{A}^*}}-\Hp{X_{\eps{A}^*}Y_{\eps{A}^*}}})/{\Hp{X_{\eps{A}^*}}}$ 
because $\sumMatrix{\eps{A}^*}>0$ for the same values of $\epsilon$ and, thus, 
\begin{equation}\label{def:Hr}
\IAs{}(A^*)=1+\lim_{\epsilon \rightarrow 0^+}\frac{\Hp{Y_{\eps{A}^*}}-\Hp{X_{\eps{A}^*}Y_{\eps{A}^*}}}{\Hp{X_{\eps{A}^*}}}=
1+\lim_{\epsilon \rightarrow 0^+}\Hr{\eps{A}^*}
\end{equation}

From Eq.\ref{eq:simple2}, Eq.\ref{eq:simple3}, and Eq.~\ref{eq:defA*Row}
we can deduce that
\begin{align*}
\THr{\eps{A}^*}&=(\ln{2})*\sumMatrix{\eps{A}^*}*\left(\Hp{Y_{\eps{A}^*}}-\Hp{X_{\eps{A}^*}Y_{\eps{A}^*}}\right)\\
&=\sum_{y=1}^n \sum_{x=1}^n \hs{\eps{A}^*[y][x]} -
\sum_{y=1}^n \hs{\sumRow{\eps{A}^*}{y}}\\
&=\left(\sum_{y=1}^m \hs{a_y}\right) + (n^2-m)*\hs{\epsilon}+\\
&\hskip1cm -\left(\sum_{y=1}^m \hs{(a_y+(n-1)*\epsilon)}\right)+\\
&\hskip1cm- (n-m)*\hs{(n*\epsilon)}\\
&=\left(\sum_{y=1}^m \hs{a_y}\right) + (n-1)*m*\hs{\epsilon}+\\
&\hskip1cm -\left(\sum_{y=1}^m \hs{(a_y+(n-1)*\epsilon)}\right)+\\
&\hskip1cm -(n-m)*(\hs{n})*\epsilon.
\end{align*}
Analogously, from Eq.~\ref{eq:simple1} and Eq.~\ref{eq:defA*Col}, it follows that:
\begin{align*}
\BHr{\eps{A}^*}&=(\ln{2})*\sumMatrix{\eps{A}^*}*\Hp{X_{\eps{A}^*}}\\
&=\sumMatrix{\eps{A}^*}*\left(\ln{\sumMatrix{\eps{A}^*}} - \frac{1}{\sumMatrix{\eps{A}^*}}*\sum_{x=1}^n \hs{\sumCol{\eps{A}^*}{x}}\right)\\%
&=\hs{\sumMatrix{\eps{A}^*}} - \sum_{x=1}^n \hs{\sumCol{\eps{A}^*}{x}}\\
&=\hs{\left((n^2-m)*\epsilon+\sum_{y=1}^m a_y\right)} +\\
&\hskip1cm - \hs{\left((n-m)*\epsilon+\sum_{y=1}^m a_y\right)} +\\
&\hskip1cm -(n-1)*\hs{n*\epsilon}\\
&=\hs{\left((n^2-m)*\epsilon+\sum_{y=1}^m a_y\right)} +\\
&\hskip1cm - \hs{\left((n-m)*\epsilon+\sum_{y=1}^m a_y\right)} +\\
&\hskip1cm -(n-1)*n*\hs{\epsilon} -(n-1)*(\hs{n})*\epsilon.
\end{align*}
Due to the continuity of $+$ and $*$ on,  
$\mathbb{R}\times \mathbb{R}$ and that of $\log$ on 
$\mathbb{R}\times \mathbb{R}_{>0}$, 
\begin{align*}
\lim_{\epsilon \rightarrow 0^+}\THr{\eps{A}^*}
&=\left(\sum_{y=1}^m \lim_{\epsilon \rightarrow 0^+}\hs{a_y}\right) + (n-1)*m*\lim_{\epsilon \rightarrow 0^+}\hs{\epsilon} +\\
&\hskip1cm -\left(\sum_{y=1}^m \lim_{\epsilon \rightarrow 0^+}\hs{(a_y+(n-1)*\epsilon)}\right)+\\
&\hskip1cm- (n-m)*(\hs{n})*\lim_{\epsilon \rightarrow 0^+}\left(\epsilon\right)\\
&=\left(\sum_{y=1}^m\hs{a_y}\right) + 0 - \left(\sum_{y=1}^m\hs{a_y}\right) - 0 - 0\\
&= 0
\end{align*}
and, in the same way,
\begin{align*}
\lim_{\epsilon \rightarrow 0^+}\BHr{\eps{A}^*}
&=\lim_{\epsilon \rightarrow 0^+}\hs{\left((n^2-m)*\epsilon+\sum_{y=1}^m a_y\right)} +\\
&\hskip0.5cm - \lim_{\epsilon \rightarrow 0^+}\hs{\left((n-m)*\epsilon+\sum_{y=1}^m a_y\right)}+\\
&\hskip0.5cm -(n-1)*n*\lim_{\epsilon \rightarrow 0^+}\hs{\epsilon}
-(n-1)*(\hs{n})*\lim_{\epsilon \rightarrow 0^+}\epsilon
\end{align*}
\begin{align*}
\lim_{\epsilon \rightarrow 0^+}\BHr{\eps{A}^*}
&=\hs{\left(\sum_{y=1}^m a_y\right)} - 
\hs{\left(\sum_{y=1}^m a_y\right)}+\\
&\hskip0.5cm -0 -0
=0.
\end{align*}

So, the limit of $\Hr{\eps{A}^*}$ cannot be directly evaluated as the ratio between the limits of $\THr{\eps{A}^*}$ and $\BHr{\eps{A}^*}$ 
because it gives rise to the indeterminate form $0/0$. 

However, if we prove that the derivative of $\BHr{\eps{A}^*}$ 
on $\epsilon$ 
is different from $0$ in a neighbourhood of $\epsilon=0$, the 
all the conditions of de l'H\^opital's rule (e.g., see~\cite{taylor52,hass2017thomas}) will be satisfied and, by the same rule, if 
\begin{align*}
\lim_{\epsilon \rightarrow 0^+}\left({\frac{\partial \THr{\eps{A}^*}}{\partial \epsilon}}\left({\frac{\partial \BHr{\eps{A}^*}}{\partial \epsilon}}\right)^{-1}\right) \in \mathcal{R}\cup \{-\infty, +\infty\}
\end{align*}
will exist, then
\begin{align}\label{eq:delhop1}
\lim_{\epsilon \rightarrow 0^+}\Hr{\eps{A}^*}
&=\lim_{\epsilon \rightarrow 0^+}\left({\frac{\partial \THr{\eps{A}^*}}{\partial \epsilon}}\left({\frac{\partial \BHr{\eps{A}^*}}{\partial \epsilon}}\right)^{-1}\right)
\end{align}
Thus, we will first compute the derivative of 
$\BHr{\eps{A}^*}$ on $\epsilon$ and, then, the limit for it as 
$\epsilon$ tends to $0$ from the right; 
if the latter exists and differs from $0$, 
then we will know that there exists a right-neighbourhood of $\epsilon=0$ 
such that its image through the derivative of 
$\BHr{\eps{A}^*}$ on $\epsilon$ does not contain $0$ and 
we can apply the de l'H\^opital's rule.

The the derivative of $\BHr{\eps{A}^*}$ on $\epsilon$ is:
\begin{align*}
\frac{\partial \BHr{\eps{A}^*}}{\partial \epsilon}&=
(n^2-m)*\ln{\left((n^2-m)*\epsilon+\sum_{y=1}^{m
}{a_{y}}\right)}+n^2-m+\\
&\hskip0.5cm-(n-m)*\ln{\left((n-m)*\epsilon+\sum_{y=1}^{m
}{a_{y}}\right)}-(n-m)+\\
&\hskip0.5cm-(n-1)*n*\ln{\epsilon} -(n-1)*n-(n-1)*\hs{n}\\
&=(n^2-m)*\ln{\left((n^2-m)*\epsilon+\sum_{y=1}^{m
}{a_{y}}\right)}-(n^2-n)*\ln{\epsilon}+\\
&\hskip0.5cm-(n-m)*\ln{\left((n-m)*\epsilon+\sum_{y=1}^{m
}{a_{y}}\right)}-(n^2-n)*\ln{n}
\end{align*}
and the limit for it as $\epsilon$ tends to $0$ is:
\begin{align*}
\lim_{\epsilon \rightarrow 0^+}\frac{\partial \BHr{\eps{A}^*}}{\partial \epsilon}
&=(n^2-m)*\lim_{\epsilon \rightarrow 0^+}\ln{ \left((n^2-m)*\epsilon+\sum_{y=1}^{m
}{a_{y}}\right)}+\\
&\hskip0.5cm-(n^2-n)*\lim_{\epsilon \rightarrow 0^+}\ln{\epsilon}
-(n^2-n)*\ln{n}\\
&\hskip0.5cm-(n-m)*\lim_{\epsilon \rightarrow 0^+}\ln{ \left((n-m)*\epsilon+\sum_{y=1}^{m}{a_{y}}\right)}
= \infty,
\end{align*}
hence, we can apply the de l'H\^opital's rule. 

The derivative of $\THr{\eps{A}^*}$ on $\epsilon$ is
\begin{align*}
\frac{\partial \THr{\eps{A}^*}}{\partial \epsilon}&=0+(n-1)*m*\ln{ \epsilon}+(n-1)*m+\\
&\hskip0.5cm - \sum_{y=1}^{m}{\left((n-1)*\ln{ \left((n-1)*\epsilon+a_{y}\right)}+(n-1)\right)}+\\
&\hskip0.5cm -(n-m)*n*\ln{n}\\
&=(n-1)*m*\ln{\epsilon} -(n-m)*n*\ln{n}+\phantom{\sum_{y=1}^{m}}\\
&\hskip0.5cm - (n-1)*\sum_{y=1}^{m}{\ln{ \left((n-1)*\epsilon+a_{y}\right)}}.
\end{align*}

The two derivatives do not share any common factor and 
they cannot be simplified. 
Moreover, 
the limit for $\Hr{\eps{A}^*}$ can not be evaluated as 
the ratio between the limits of the derivatives  of $\THr{\eps{A}^*}$ 
and $\BHr{\eps{A}^*}$ because it has the form 
$-\infty/\infty$, which is indeterminate. As a matter of fact, 
\begin{align*}
\lim_{\epsilon \rightarrow 0^+}\frac{\partial \THr{\eps{A}^*}}{\partial \epsilon}
&=(n-1)*m*\lim_{\epsilon \rightarrow 0^+} (\ln{\epsilon})
-(n-m)*n*\ln{n} +\\
&\hskip0.5cm-(n-1)*\sum_{y=1}^{m}{\lim_{\epsilon \rightarrow 0^+}\ln{ \left((n-1)*\epsilon+a_{y}\right)}}
= -\infty, 
\end{align*}

Luckly, de l'H\^opital's rule can be applied again because 
 the second derivative of $\BHr{\eps{A}^*}$ on $\epsilon$ is:
\begin{align*}
\frac{\partial^2 \BHr{\eps{A}^*}}{\partial \epsilon^2}&=
\frac{(n^2-m)^2}{(n^2-m)*\epsilon+\sum_{y=1}^{m}{a_{y}}}
-\frac{n^2-n}{\epsilon} +\\
&\hskip0.5cm
-\frac{(n-m)^2}{(n-m)*\epsilon+\sum_{y=1}^{m
 }{a_{y}}}\\
&=-\frac{(n^2-n)*\sum_{y=1}^m a_{y}}{\epsilon*\left((n^2-m)*\epsilon+\sum_{y=1}^{m}{a_{y}}\right)*\left((n-m)*\epsilon+\sum_{y=1}^{m
	}{a_{y}}\right)},
\end{align*}
and the limit for it as $\epsilon$ tends to $0$ is: 
\begin{align*}
\lim_{\epsilon \rightarrow 0^+}\frac{\partial^2 \BHr{\eps{A}^*}}{\partial \epsilon^2}&=
\lim_{\epsilon \rightarrow 0^+}\frac{(n^2-m)^2}{(n^2-m)*\epsilon+
	\sum_{y=1}^{m}{a_{y}}}
-\lim_{\epsilon \rightarrow 0^+}\frac{(n^2-n)}{\epsilon} +\\
&\hskip1cm
-\lim_{\epsilon \rightarrow 0^+}\frac{(n-m)^2}{(n-m)*\epsilon+\sum_{y=1}^{m
	}{a_{y}}}\\
&=\frac{(n^2-m)^2}{\sum_{y=1}^{m}{a_{y}}}-\infty
-\frac{(n-m)^2}{\sum_{y=1}^{m}{a_{y}}}\\
&= -\infty.
\end{align*}
So, there exists a right-neighbourhood of $\epsilon=0$ such that 
none of its values is mapped in $0$ through the second derivative 
of $\BHr{\eps{A}^*}$.

The ratio between 
$\partial^2 \THr{\eps{A}^*}/\partial \epsilon^2$ and 
$\partial^2 \BHr{\eps{A}^*}/\partial \epsilon^2$ can be 
algebraically simplified because they both have $1/\epsilon$ as a 
factor. As a matter of fact,
\begin{align*}
\frac{\partial^2 \THr{\eps{A}^*}}{\partial \epsilon^2}&=
\frac{(n-1)*m}{\epsilon}-(n-1)*\sum_{y=1}^m\frac{n-1}{(n-1)*\epsilon+a_y}\\
&=\frac{(n-1)*m-\epsilon*\sum_{y=1}^m\frac{(n-1)^2}{(n-1)*\epsilon+a_y}}{\epsilon}
\end{align*}
and
\begin{align}
\frac{\frac{\partial^2 \THr{\eps{A}^*}}{\partial \epsilon^2}}
{\frac{\partial^2 \BHr{\eps{A}^*}}{\partial \epsilon^2}} &= - 
{\frac{(n-1)*m-\epsilon*\sum_{y=1}^m\frac{(n-1)^2}{(n-1)*
\epsilon+a_y}}{\epsilon}}*\nonumber\\
&\hskip1cm *{\frac{\epsilon*\left((n^2-m)*\epsilon+\sum_{y=1}^{m}{a_{y}}\right)*\left((n-m)*\epsilon+\sum_{y=1}^{m
		}{a_{y}}\right)}{(n-1)*n*\sum_{y=1}^m a_{y}}}\nonumber\\
&=-\frac{m*(n-1)*\sum_{y=1}^m{a_y}}{n*(n-1)*\sum_{y=1}^m a_{y}}+\label{eq:derTopBot2}\\
&\hskip0.5cm-\epsilon^2*\frac{(n-1)*m*(n^2-m)*(n-m)}{(n-1)*n*\sum_{y=1}^m a_{y}} +\nonumber\\
&\hskip0.5cm-\epsilon*\frac{(n-1)*m*(n^2-m)*\sum_{y=1}^{m}{a_{y}}}{(n-1)*n*\sum_{y=1}^m a_{y}} +\nonumber\\
&\hskip0.5cm-\epsilon*\frac{(n-1)*m*\left(\sum_{y=1}^{m}{a_{y}}\right)*(n-m)}{(n-1)*n*\sum_{y=1}^m a_{y}} +\nonumber\\
&\hskip0.5cm-\epsilon^2*\frac{(n-1)*m*(n^2-m)*(n-m)}{(n-1)*n*\sum_{y=1}^m a_{y}} +\nonumber\\
&\hskip0.5cm+\epsilon^3*\frac{\left(\sum_{y=1}^m\frac{(n-1)^2}{(n-1)*
		\epsilon+a_y}\right)*(n^2-m)*(n-m)}{(n-1)*n*\sum_{y=1}^m a_{y}} + \nonumber\\
&\hskip0.5cm+\epsilon^2*\frac{\left(\sum_{y=1}^m\frac{(n-1)^2}{(n-1)*
		\epsilon+a_y}\right)*\left(\sum_{y=1}^{m}{a_{y}}\right)*(n-m)}{(n-1)*n*\sum_{y=1}^m a_{y}} + \nonumber\\
&\hskip0.5cm+\epsilon^2*\frac{\left(\sum_{y=1}^m\frac{(n-1)^2}{(n-1)*
		\epsilon+a_y}\right)*(n^2-m)*\sum_{y=1}^{m}{a_{y}}}{(n-1)*n*\sum_{y=1}^m a_{y}}\nonumber\\
&\hskip0.5cm+\epsilon*\frac{\left(\sum_{y=1}^m\frac{(n-1)^2}{(n-1)*
		\epsilon+a_y}\right)*\left(\sum_{y=1}^{m}{a_{y}}\right)^2}{(n-1)*n*\sum_{y=1}^m a_{y}}\nonumber
\end{align}

The first term of Eq.~\ref{eq:derTopBot2} equals $-m/n$, while 
each of the remaining terms has instead the form 
\[
\epsilon^{c}*\frac{\left(\sum_{y=1}^m\frac{(n-1)^2}{(n-1)*
		\epsilon+a_y}\right)^d*p(n,m,a_1,\ldots,a_m)}{(n-1)*n*\sum_{y=1}^m{a_y}}
\] 
for suitable natural numbers $c\in \{1,2,3\}$ and $d\in \{0,1\}$ and 
fitting polynomial function 
$p(n,m,a_1,\ldots,a_m)$. 
Since $p(n,m,a_1,\ldots,a_m)$ is constant with respect $\epsilon$ and, 
under the assumptions we made for $\eps{A}^*$, 
$(n-1)*n*\sum_{y=1}^m{a_y}$ is a positive real value, 
it is easy to see that the limit as $\epsilon$ tends 
to $0$ for each of the terms of Eq.~\ref{eq:derTopBot2}, 
but the first one, is $0$. It follows that, 
\begin{equation}\label{eq:finalHr}
\lim_{\epsilon \rightarrow 0^+}\Hr{\eps{A}^*}=\lim_{\epsilon \rightarrow 0^+}\frac{\frac{\partial^2 \THr{\eps{A}^*}}{\partial \epsilon^2}}
{\frac{\partial^2 \BHr{\eps{A}^*}}{\partial \epsilon^2}}=-\frac{m}{n}
\end{equation}
and the following theorem holds.
\begin{thm}\label{theo:hard}
Let ${\gAM{}}$ be an $n\times n$-agreement matrix.
%
If $\Hp{\nnrv{X_{\gAM{}}}}=0$ and ${\gAM{}}$ accounts 
exactly $m$ non-null rows,
then $\IAs{}({\gAM{}})$ exists and it equals $(n-m)/n$.
\end{thm}
\begin{proof}
The proof directly follows from both Eq.~\ref{def:Hr} and Eq.~\ref{eq:finalHr}.
\end{proof}

Since, whenever defined, $\IA{}$ is symmetric with respect to
transposition, i.e., $\IA{\gAM{}}=\IA{\trans{\gAM{}}}$, we can 
prove the following corollary.

\begin{cor}\label{cor:full}
Let ${\gAM{}}$ be an $n\times n$-agreement matrix.
The information agreement extension by continuity of $B$,  
$\IAs{}({\gAM{}})$, does exist.
Moreover, if $l$ and $m$ are numbers of 
non-null columns and non-null rows in $B$, 
respectively, then 
\begin{equation}\def\arraystretch{2}
\IAs{}({\gAM{}}) = \left\{\begin{array}{lll}
\frac{n-l}{n}&\mbox{}&\text{if $\Hp{\nnrv{Y_{\gAM{}}}}=0$}\\
\frac{n-m}{n}&&\text{if $\Hp{\nnrv{X_{\gAM{}}}}=0$}\\
1+\frac{\Hp{\nnrv{Y_{\gAM{}}}}-\Hp{\nnrv{X_{\gAM{}}}\nnrv{Y_{\gAM{}}}}}{\Hp{\nnrv{X_{\gAM{}}}}}&&
\text{if $0<\Hp{\nnrv{X_{\gAM{}}}}\leq \Hp{\nnrv{Y_{\gAM{}}}}$}\\
1+\frac{\Hp{\nnrv{X_{\gAM{}}}}-\Hp{\nnrv{X_{\gAM{}}}\nnrv{Y_{\gAM{}}}}}{\Hp{\nnrv{Y_{\gAM{}}}}}&&
\text{if $0<\Hp{\nnrv{Y_{\gAM{}}}}\leq \Hp{\nnrv{X_{\gAM{}}}}$}
\end{array}\right.
\end{equation}
\end{cor}
\begin{proof}
The proof of the claim directly follows from Lemma~\ref{lem:transpose_nonnull}, Theorem~\ref{theo:easy}, and Theorem~\ref{theo:hard}. 
\end{proof}

%% file: computing.tex

\section{Computing $\IAs{}$}

\DontPrintSemicolon
\SetKwData{nnrows}{m}
\SetKwData{nncols}{l}
\SetKwData{nrows}{n}
\SetKwData{row}{row}
\SetKwData{HXR}{HX\_R}
\SetKwData{HYR}{HY\_R}
\SetKwData{HXYR}{HXY\_R}
\SetKwData{acc}{acc}
\SetKwData{pres}{res}
\SetKwData{sumA}{S\_A}

\SetKwFunction{getH}{H}
\SetKwFunction{getsumA}{get\_S\_A}
\SetKwFunction{refine}{refine}
\SetKwFunction{getpY}{get\_pY}
\SetKwFunction{getpX}{get\_pX}
\SetKwFunction{getpXY}{get\_pXY}
\SetKwFunction{countnncols}{countNonNullCols}
\SetKwFunction{countnnrows}{countNonNullRows}
\SetKwFunction{size}{size}
\SetKwFunction{IAC}{getIAC}

\SetKwProg{Fn}{def}{\string:}{}
\SetKwInOut{Input}{Input}\SetKwInOut{Output}{Output}

Corollary~\ref{cor:full} not only guarantees the existence of 
$\IAs(A)$ for 
any agreement matrix $A$, but also provides an effective way to compute 
it. Algorithm~\ref{alg:IAC} is the algorithmic counterpart of 
Corollary~\ref{cor:full} and the correctness of the former 
follows directly from the latter. 

As far as the complexity of Algorithm~\ref{alg:IAC} may concern, 
line~\ref{line:getN} can certainly be assumed to take constant 
time with respect to the size of $A$.
It is easy to figure out that 
lines~\ref{line:pX}, \ref{line:pY}, and  \ref{line:pXY}, which 
compute $\Hp{\nnrv{X_{A}}}$, $\Hp{\nnrv{Y_{A}}}$, and, 
$\Hp{\nnrv{X_{A}}\nnrv{Y_{A}}}$, respectively, take time 
$\Theta(n^2)$, i.e, their execution times are upper-bounded and 
lower-bounded by functions  
proportional to $n^2$ in both best and worst-case scenarios (e.g., see~\cite{cormen01introduction}).
If $A$ is an $n\times n$ matrix, 
then both lines~\ref{line:nnrows} and~\ref{line:nncols} take time 
$O(n^2)$, i.e, in the worst-case scenario, 
their execution times are upper-bounded by functions 
proportional to $n^2$ (e.g., see~\cite{cormen01introduction}).
All the remaining lines take constant time with respect to the input size.
So, the overall cost of Algorithm~\ref{alg:IAC} is $\Theta(n^2)$.

\setcounter{algocf}{1}
\begin{algorithm}[H]
	\Input{A generic agreement matrix $A$}
	\Output{The value $\IAs(A)$}
	
	\BlankLine
	\Fn{\IAC{$A$}}{
		\nrows $\leftarrow$ $A.\size$\label{line:getN}\tcc*[r]{get the number of rows/cols in $A$}
		\BlankLine
		
		\HXR $\leftarrow$
		\getH(\refine(\getpX($A$)))\label{line:pX}\tcc*[r]{compute  $\Hp{\nnrv{X_{A}}}$}
		
		\If{\HXR$=0$\label{line:ifhard1}}{
			\nnrows $\leftarrow$ \countnnrows($A$)\label{line:nnrows}\tcc*[r]{count the non-null rows}
			\BlankLine
			\KwRet{$(\nrows-\nnrows)/\nrows$}\label{line:return1}\;
		}
	
	\BlankLine
		\HYR $\leftarrow$
		\getH(\refine(\getpY($A$)))\label{line:pY}\tcc*[r]{compute  $\Hp{\nnrv{Y_{A}}}$}

		\If{\HYR$=0$\label{line:ifhard2}}{
			\nncols $\leftarrow$ \countnncols($A$)\label{line:nncols}\tcc*[r]{count the non-null cols}
			\BlankLine
			\KwRet{$(\nrows-\nncols)/\nrows$}\label{line:return2}\;
		}		
		\BlankLine
		\HXYR $\leftarrow$ \getH(\refine(\getpXY($A$)))\label{line:pXY}\tcc*[r]{compute  $\Hp{\nnrv{X_{A}}\nnrv{Y_{A}}}$}
		\eIf{$\HXR<\HYR$\label{line:ifeasy}}{
			\KwRet{$1+(\HYR-\HXYR)/\HXR$}\;
		}{
			\KwRet{$1+(\HXR-\HXYR)/\HYR$}\label{line:return4}\;
		}
	}
	\caption{Computes $\IAs(A)$ for any 
		agreement matrix $A$.}\label{alg:IAC}
\end{algorithm}